\documentclass{cta-author}

{}
{}
\usepackage{algorithmic}
\usepackage{graphicx}
\usepackage{textcomp}
\usepackage{braket}
\usepackage[ruled,vlined]{algorithm2e}
\usepackage{graphicx}
\usepackage{makecell}
\usepackage{caption}
\usepackage{amsthm}
\newtheorem{theorem}{Lemma}
\begin{document}

\supertitle{IET Quantum Communication}

\title{A Novel Quantum Algorithm for Ant Colony Optimization}

\author{\au{Mrityunjay Ghosh$^{1,3\corr}$}, \au{Nivedita Dey$^{2,3}$}, \au{Debdeep Mitra$^{2,3}$}, \au{Amlan Chakrabarti$^{3}$}}

\address{
\add{1}{HCL Technologies, India}
\add{2}{QRDLab, India}
\add{3}{University of Calcutta, India}
\email{g.mrityunjay@gmail.com}}

\begin{abstract}
Ant colony optimization (ACO) is a commonly used meta-heuristic to solve complex combinational optimization problems like traveling salesman problem (TSP), vehicle routing problem (VRP), etc. However, classical ACO algorithms provide better optimal solutions but do not reduce computation time overhead to a significant extent. Algorithmic speed-up can be achieved by using parallelism offered by quantum computing. Existing quantum algorithms to solve ACO are either quantum-inspired classical algorithms or hybrid quantum-classical algorithms. Since all these algorithms need the intervention of classical computing, leveraging the true potential of quantum computing on real quantum hardware remains a challenge. This paper's main contribution is to propose a fully quantum algorithm to solve ACO, enhancing the quantum information processing toolbox in the fault-tolerant quantum computing (FTQC) era. We have Solved the Single Source Single Destination (SSSD) shortest-path problem using our proposed adaptive quantum circuit for representing dynamic pheromone updating strategy in real IBMQ devices. Our quantum ACO technique can be further used as a quantum ORACLE to solve complex optimization problems in a fully quantum setup with significant speed up upon the availability of more qubits.
\end{abstract}

\keywords{Quantum Computing, Quantum Algorithm, Ant Colony Optimization, QACO, Quantum Ant, Quantum Circuit Synthesis}

\maketitle

\section{Introduction}

Artificial Swarm Intelligence (ASI) helps amplify group intelligence of a formed network to enable mere accurate forecasts, insights, evaluations, and assessments. \cite{holland} An ant colony is a natural system that, as a whole, is capable of engaging in complex behaviors like building nests, foraging for food, raising aphid livestock, waging war with other colonies, etc. The concept of Ant Colony Optimization (ACO) was first proposed by an Italian scholar, M.Dorigo, in 1991. The foraging behavior of real-world ant colonies can be mapped as a meta-heuristic for solving discrete combinatorial optimization problems. \cite{dorigo} \cite{colorni} Real ants release a chemical on their path from nest to food called pheromone, which is used by following ants to find the shortest path to food source via the pheromone trail. Ant colony metaphor can solve a diverse range of optimization problems like Travelling Salesman Problem (TSP), \cite{tsp} Vehicle Routing problem, \cite{dants} Assignment problem, \cite{demiral} Job-shop problem, \cite{shopscheduling} 0-1 Knapsack \cite{zhao} and many more. ACO has proven itself to be a promising one by exhibiting tremendous growth in solving discrete optimization problems. \cite{blum} But, synthesis of an optimal reversible circuit for a given optimization function is an NP-Hard problem. Quantum computing principles like superposition, entanglement, and quantum tunneling can bring a paradigm shift in computing by achieving substantial speed-up over its classical counterpart.\cite{hoos}\cite{kallel} Quantum computing offers reversible computational logic, which helps in the minimization of power consumption. Reversible quantum circuits synthesized with basic quantum gates like NOT, CNOT, TOFFOLI do not lose computational information and thus can be a good alternative for reducing overall computational complexity and resource overhead. Additionally, quantum superposition using Hadamard gate will allow parallel exploitation of the search space and thus may offer the desired solution faster than classical ACO synthesis through quantum interference. \cite{grover}

Evolutionary computation (EC) is a class of heuristic optimization techniques inspired by biological evolution to solve challenging hard problems. In EC, solution space is a population of candidate solutions that compete against each other to gradually increase the population's fitness from one generation to the next using operators like mutation, crossover, cloning, and intelligent selection strategies. After the evolution process, highly fit individuals will emerge as optional solutions in the solution space. A quantum-inspired evolutionary algorithm was the first genre of solutions to connect quantum computing with evolutionary computing, which was designed to use quantum logic to inspire the creation of better optimization algorithms that can be run on classical computers. \cite{narayanan} Next attempt in this space was to incorporate a class of evolutionary computing known as Genetic computing for the evolution of new quantum algorithms. These quantum-inspired algorithms do not explain how classical representations and operators can be implemented using quantum gates, how the entanglement of qubits can be manifested, or how quantum interference can yield classical measurement of a globally optimal solution. Moreover, quantum genetic algorithms (QGA) \cite{hunkim} \cite{rylander} made use of quantum gates for rotation to implement evolution function but cannot use mutation or crossover operator. This is because any prior observation or measurement of quantum states would destroy the superposition. Moreover, the no-cloning theorem of quantum mechanics prevents the design of quantum mutation operators. Quantum evolutionary algorithms (QEAs) proposed in \cite{wang-qa}, \cite{wang2007} adopted quantum crossover operation proposed by  Narayanan and Moore \cite{narayanan} and provided equations to guide chromosome update classically. These provide clear evidence of the shortcomings of quantum-inspired evolutionary techniques for implementing on real quantum hardware.

Near-term applications on a quantum computer are primarily based on a hybrid quantum-classical variational approach. These approaches work based upon the parameterization of quantum states applied on a relatively small parameter set. Classical optimization modules determine parameter values based on optimization of the utility function, which is nothing but a Hamiltonian encoding of the total energy of the underlying system. But these hybrid variational approaches to solve classical non-linear optimization problems are subjected to several performance issues, including hardness of finding Hamiltonians with many more Pauli terms, difficulty in avoiding multiple local optima, etc. Exploiting quantum mechanical principles will require a full-scale quantum computer with noiseless and sufficient qubits for computation. Whenever these fault-tolerant quantum computers (FTQC) will be available, the logical error rate associated with qubits will be suppressed to arbitrarily low levels. Then quantum computers will perform long computations with sufficient noise-tolerant qubits to produce better outcome precision. Our proposed QACO algorithm is designed for the FTQC era, which can be added in the FTQC toolbox as a fully quantum module to solve the Single Source Single Destination (SSSD) problem using principles of the basic ant colony system.

We have attempted to introduce quantum programming techniques like amplitude amplification and modified Grover ORACLE operation to provide complex quantum subroutines which can be fitted for an increasing variety of specialized data structures. A real challenge in developing a full-scale quantum algorithm was to generate fitness values using a quantum ORACLE to mimic the behavior of evolutionary algorithms without any classical intervention. Transformation of superposed quantum states representing dynamic update of pheromone using circuit model of computation was one of the most significant computationally challenging parts. Our proposed approach shows how individuals encoded in superposed quantum states can maintain a correlation between the process of initial state preparation and quantum pheromone box transformation. Our proposed method does not show any computational advantage over the classical counterpart with the same linear-time complexity. But, if our proposed QACO module can be integrated as an ORACLE to solve NP-Hard problems, it can surely give new direction for future quantum computation in exploiting quantum computing supremacy.

\begin{itemize}
\item \textbf{Ant Colony Optimization Algorithm:}
The idea invented so far regarding Quantum Ant Colony Optimization (QACO) is based on QEA, where Q-bit and Quantum Rotation strategy are used to represent and update the pheromone respectively in discrete binary combinatorial optimization domain. In this paper, we have presented a novel quantum ant colony optimization algorithm based on hybrid implementation of Grover\textquotesingle{s} amplitude amplification technique \cite{grover} and time driven quantum evolution to encode the possible paths by updating the distribution of pheromone in terms of probability density function. The amount of pheromone deposited on the shortest path will eventually lead to global convergence of the optimization problem in the solution space. The implementation and simulation results exhibit better efficiency, improved computation speed and enhanced optimization capability of ACO making it a robust solution. To the best of our knowledge, the QACO proposed in this paper is the first attempt towards a full scale quantum technique where all the basic operations of Ant Colony Optimization like path exploration, pheromone deposition and pheromone evaporation are performed in quantum registers.
\end{itemize}
The rest of this paper starts with literature review section as Section 2. Section 3 covers the basics of combinatorial optimization (CO) problems and classical ACO to solve CO problems along with its variations, section 4 illustrates quantum gates and circuits as prerequisites in our proposed MNDAS (Mrityunjay-Nivedita-Debdeep-Amlan-Subhansu) algorithm. Section 5 sheds light on conventional quantum-inspired evolutionary technique for ACO, section 6 elucidates our novel QACO algorithm along with its problem initialization, path exploration and pheromone updation modules. Section 7 provides a mapping between the proposed algorithm and QACO problem in terms of realizing the approach on an ant colony example and presents simulation result showing global path convergence and time complexity analysis. The last section illustrates conclusion and scope for future research growth on related areas. \cite{kallel}

\section{Literature Review} ACO has a diverse range of applications and thus has always been a fundamental topic of theoretical interest. Since our work aims to propose a novel 'quantum algorithm' for ACO, we will mention fewer insights into working principles of classical ACO through their evolution and highlight the notion of quantum ant colony optimization based Quantum-inspired Evolutionary Algorithms (QEAs).
During evolution starting from classical ACO metaheuristic by M. Dorigo, the first rigorous theoretical investigations on ACO was proposed by Neumann and Witt in 2006, where authors had presented a $1$-ANT algorithm. \cite{newmannwitt} The algorithm operates by constructing a new solution and performing pheromone update only if the current solution is better than the best solution obtained so far. Pheromone update in $1$-ANT algorithm is controlled by evaporation factor ($\rho$) [$0<\rho<1$]. The larger value of $\rho$ is associated with increasing the impact factor of the current solution over the previous best-obtained solution. Later, investigations on $1$-ANT ACO have shown the performance degradation of the algorithm for a minimal value chosen for evaporation factor, as the expected time to achieve an optimal solution is exponential. \cite{newmannwitt} \cite{doerr} Further work on classical ACO was made by reinforcing the best solution obtained so far in each iteration using the best-so-far update strategy. This concept was first coined in MAX-MIN Ant System (MMAS) algorithm. \cite{doerr} \cite{gutjahr} \cite{mmhoos} Another algorithmic advancement was made in this direction by reinforcing the best solution created in the current iteration. This is known as an iteration-best update that works well even with a small value of evaporation factor. \cite{sudholt-ising}
Transformation of a discrete optimization problem as a `best path' problem is the intuition behind the origination of quantum ant colony optimization, which is inspired by Quantum-inspired Evolutionary Algorithms (QEAs) theories. Unlike classical ACO, QEA based ACO is formulated with the help of Q-bit and quantum rotation gate. As the pheromone update strategy is dissimilar to the update strategy of QEA, Ling Wang et al. first proposed a rotation angle updating strategy to update the pheromone trails over the existing pre-determined updating strategy in 2007. \cite{wang2007}  Their proposed work significantly taken into consideration the exploitation probability, as optimizing exploitation probability provides a trade-off between earlier convergence of ACO and effective escape from local optima. In 2008, Ling Wang et al. extended their algorithm to solve fault detection in the chemical production process. \cite{wang-qa} The authors had combined Support Vector Machine (SVM) with their proposed QACO to select fault features. In 2010, Panchi Li et al. proposed a continuous quantum ant colony optimization algorithm. They have made each ant encode with a group of qubits to represent its own position.\cite{li-qa} The algorithm begins with selecting the local best path based on pheromone information as a heuristic followed by updating each of the ant\textquotesingle{s} own qubits with the help of quantum rotation gate. In order to enable mutation and improve the diversity of positions, some qubits have undergone modifications by quantum non-gates. The idea of continuous-time evolution was implemented by adding the fitness function value of the current ant position of the pheromone to update the heuristic information. The optimum position thus can hold the more excellent fitness function value and fitness function gradient value ensuring accelerated, guaranteed convergence. Some other work in the domain of QACO was that of evacuation path optimization algorithm proposed by Min Liu et al. in 2016. \cite{evacuationpath} In this paper, the advantage lies upon the scalability of the method as it is suitable for multiple source nodes to multiple destination nodes instead of a single path between two locations.

\section{Classical Ant Colony Optimization for solving combinatorial optimization problems}
The objective of a combinatorial optimization problem, associated with a set of problem instances, is to maximize or minimize several parameters. \cite{combinatorialoptimization} The solutions of these intractable optimization problems incur exhaustive or brute-force search, which is computationally hard. \cite{blum} Ant Colony Optimization (ACO) is a meta-heuristic process to solve computationally hard optimization problems, in which the idea is to allocate the computational resources to a set of relatively simple agents called artificial ants. Since, in our paper we have focused on proposing quantum algorithm for simple ACO, where path searching behavior of ants and pheromone updation rules are discussed, we will have a brief introduction to simple ACO meta-heuristic.
Simple ACO functions in two operating modes: forward (from nest towards the food) and backward (from food back to the nest). Forward ants build a solution by probabilistically choosing the next node to move to among those in the adjacent positions with respect to current node. This probabilistic choice is biased by pheromone trails previously deposited on the paths by other ants. Forward ants do not deposit any pheromone, which when associated with deterministic backward moves, helps to eliminate loop formation.
\subsection{Ants' path searching behavior}
Each ant builds a solution to the problem, where step-by-step decision making of each ant relies on the local information stored on the node itself or on the outgoing arcs. Search process begins with assignment of constant amount of pheromone $(\tau_{ij}=1)$ to all arcs. \cite{dorigo} \cite{colorni} An ant $k$ locating on node $i$ computes the probability to choose $j$ as the next node using pheromone trail $\tau_{ij}$ as follows:
\begin{equation}
{p_{ij}}^{k}=
\begin{cases}
      \frac{\tau_{ij}^{\alpha}}{\sum\limits_{j \in N_{i}^{k}}\tau_{ij}^{\alpha}}, & \text{if $j \in N_{i}^{k}$}\\
      0, & \text{if $j \notin N_{i}^{k}$}\\
    \end{cases}
\end{equation}
$N_{i}^{k}$\:=\:Neighbourhood of ant $k$ when located on node $i$ (excluding predecessor of node $i$)
\\\\
In ACO meta-heuristic, a problem-specific heuristic is taken into account for decision making. \cite{dorigo1} \cite{dorigo2}
\begin{equation}
f(j)=
\begin{cases}
	arg\{\begin{matrix}max\\d=feasible_{k}(t)\end{matrix}[\tau_{id}(t)^{\alpha}.\eta_{id}^{\beta}]\:\} & \text{when $r \leq r_{0}$}\\
	j^{'} & \text{when $r>r_{0}$}
\end{cases}
\end{equation}
$f(j)=\:$ Constraint for transition function to move from node $i$ to $j$\\\\
$\tau_{id}(t)=\:$ Pheromone trail at time $t$\\
$\eta_{id}=\:$ problem specific heuristic information\\
$\alpha=\:$impact of heuristic information\\
$r=\:$random number with uniform distribution in $[0,1]$\\
$r_{0}=\:$pre-specified parameter ranging from $0$ to $1$, inclusive\\\\
$feasible_{k}(t)=\:$set of feasible nodes excluding already visited (predecessor) nodes by $k$-th ant before visiting node $i$, to prevent loop formation\\\\
$j^{'}=\:$target point selected according to the following probability distribution.\\
\begin{equation}
P_{ij}^k(t)=
\begin{cases}
	\frac{[\tau_{ij}(t)]^{\alpha} . \eta_{ij}^{\beta}}{\sum\limits_{d \in feasible_{k}(t)}[\tau_{ij}(t)]^{\alpha} . \eta_{ij}^{\beta}}, & \text{if $j \in feasible_k(t)$}\\
	0, & \text{otherwise}
\end{cases}
\end{equation}
\subsection{Deterministic backward pheromone trail update}
Retracing step by step by the same path in backward mode first begins with a scanning process for formed loop elimination followed by deposition of $\Delta\tau^{k}$ amount of pheromone by each ant. An ant $k$ traversing in backward mode through the arc $(i,j)$, will update the pheromone value as follows:
\begin{equation}
	\tau_{ij}^{'}=\tau_{ij}+{\Delta\tau}^{k}
\end{equation}
This pheromone updation step ensures chances of forthcoming ants to trace the same path.
\subsection{Pheromone trail evaporation}
Pheromone trail evaporation is an exploration mechanism to avoid quick convergence of all ants towards a local best solution, equivalently a suboptimal path. Process of decreasing the intensities of pheromone trails favours exploration of different paths during whole search space. The evaporation of pheromone trails can be expressed as follows:
\begin{equation}
\tau_{ij}^{*}\leftarrow(1-\rho).{\tau_{ij}}\ \ \ \ \forall \: (i,j)\in E(G)
\end{equation}
$E(G)$ represents the set of all arcs of the graph, $\rho \in (0,1]$ is a parameter, ${\tau_{ij}}^{*}$ is the updated pheromone level after evaporation.\cite{dorigo} \cite{colorni} \cite{dorigo3}
\subsection{Daemon-actions}
The pheromone evaporation process is interleaved with the process of pheromone deposition with $\Delta\tau^{k}$ amount to all arcs. But at times, activating a local optimization procedure to implement centralized actions, is an important factor to decide utility of depositing additional pheromone to bias the search process from a non-local perspective. \cite{colorni} \cite{dorigo} \cite{evacuationpath}\\
Evolution in simple ACO has taken place with difference in updation policy of pheromone deposition and evaporation. Ant-System(AS) update rule of simple ACO is replaced by Iteration-Best (IB) update rule in practice. IB update rule introduces emergence into a system while taking biasedness towards good solutions into consideration obtained through previous iterations. Another variant is Best-so-far (BS) update rule, which exhibits biasedness towards the best solution available so far. \cite{blumdorigo}\cite{blumdorigo1} Both of these policies suffer from earlier convergence as the set of all sequences of solution components that might be constructed by ACO algorithm to produce feasible solutions are updated with the set of all sequences of solution components obtained through multiple solutions of previous iterations in case of IB update and best solution obtained so far in case of BS update. In order to avoid premature convergence, advanced ACO algorithms like Ant Colony System (ACS) and MAX\_MIN Ant System (MNAS) are used.
\section{Quantum Gates and Circuits} The smallest unit of information in quantum is represented as quantum bits or qubits. A qubit is thought to exist as a superposition of two pure states $0$ and $1$. A state of a superposed qubit can be expressed as follows:
\begin{equation}
	\ket{\psi}=\alpha\ket{0}+\beta\ket{1}
\end{equation} Here, $\alpha$ and $\beta$ are complex numbers with $|\alpha|^{2}$ and $|\beta|^{2}$ representing probabilistic amplitudes of the superposed qubit to be in $\ket{0}$ and $\ket{1}$ respectively. \cite{lovett}
Unlike classical computing with irreversible gate logic, quantum computing performs unitary evolution of quantum states and hence relies on reversible logic of quantum gates. \cite{lovett} A single qubit gate is a kind of operator that acts on only one qubit at a time. These operators are described by $2 \times 2$ unitary matrices, where unitary matrix $U$ has a property $U^\dagger.U=U.U^\dagger=I$.\cite{lovett} \cite{grover} Such an operator, called Hadamard operator or Hadamard gate maps the basic state $\ket{0}$ to $\frac{\ket{0}+\ket{1}}{\sqrt{2}}$ and $\ket{1}$ to $\frac{\ket{0}-\ket{1}}{\sqrt{2}}$. Hadamard gate is an 1-qubit version of QFT (Quantum Fourier Transform). \cite{lovett}
\begin{multline}
H=\frac{1}{\sqrt{2}}\begin{pmatrix}1 & 1 \\ 1 & -1\\ \end{pmatrix},\\ H\ket{0}=\frac{1}{\sqrt{2}}\begin{pmatrix}1 & 1 \\ 1 & -1\\ \end{pmatrix}\begin{pmatrix}1 \\ 0 \end{pmatrix} =\frac{1}{\sqrt{2}}\begin{pmatrix}1 \\ 1 \end{pmatrix}=\frac{1}{\sqrt{2}}(\ket{0}+\ket{1})
\end{multline}
A quantum gate can act on $N$ qubits simultaneously. Similar to the case of single qubit, the probability must be conserved when operating in multiple dimensions, and the operators are hence unitary. The simplest example is the well-known two qubit Controlled $NOT$ ($CNOT$) gate or Fenyman gate. Matrices are defined in the basis spanned by the two qubit state vectors $\ket{00} \equiv [1000]^T$, $\ket{01} \equiv [0100]^T$, $\ket{10} \equiv [0010]^T$, $\ket{11} \equiv [0001]^T$, where the first qubit is the control qubit and the second qubit is the target qubit. The $CNOT$ gate flips the state of the target qubit conditioned on the control qubit being in state $\ket{1}$. The action of the $CNOT$ gate is given as $\ket{x}\ket{y} \to \ket{x}\ket{y \oplus x}$. \cite{lovett}
\begin{center}
	$CNOT$ = $\begin{pmatrix}
		1 & 0 & 0 & 0 \\
		0 & 1 & 0 & 0 \\
		0 & 0 & 0 & 1 \\
		0 & 0 & 1 & 0 \\
	\end{pmatrix}$
\end{center}
Controlled-$NOT$ gate can be extended to Controlled-Controlled-$NOT$ ($C^2NOT$) gate, alternatively known as Toffoli gate. It is universal reversible logic acting as a quantum operator with three input bits. If first two bits are set to 1, it inverts the third bit. \cite{lovett}
Controlled phase ($CPHASE$) gate applies a $Z$-gate to the target qubit conditioned on the control qubit being in state $\ket{1}$. $Z$ gate performs a $\pi$- rotation around the $Z$-axis. $Z$-gate is also referred as phase flip. \cite{lovett} 
\begin{center}
	$CPHASE$ = $\begin{pmatrix}
		1 & 0 & 0 & 0 \\
		0 & 1 & 0 & 0 \\
		0 & 0 & 1 & 0 \\
		0 & 0 & 0 & e^{i\theta} \\
	\end{pmatrix}$
\end{center}
A Controlled-Controlled $NOT$ or ($C^2NOT$) gate gives the AND of two control qubits $C_1$ and $C_2$. Chaining more than two Toffoli\textquotesingle{s} together through AND operation among multiple control qubits $c_1, c_2, ..., c_n (c_1.c_2.c_3...c_n)$ and introducing few ancilla qubits to store intermediate results, $C^nNOT$ gate can be implemented. In figure \ref{fig:tofolli4}, a 4 qubit Toffoli ($C^4NOT$) as MCT has been implemented with four control qubits $\ket{c_1}$, $\ket{c_2}$, $\ket{c_3}$, $\ket{c_4}$ and three ancilla qubits and one target qubit which will flip only when $c_1.c_2.c_3.c_4 = 1$. After applying the final $C^n(X)$, a reversible model of computation has been implemented (compute-copy-uncompute) to clean up intermediate work qubits by undoing their computation and resulting ancilla qubits to $\ket{0}$ state. \cite{gupta}
\begin{figure}
  \includegraphics[width=8cm , height=5cm]{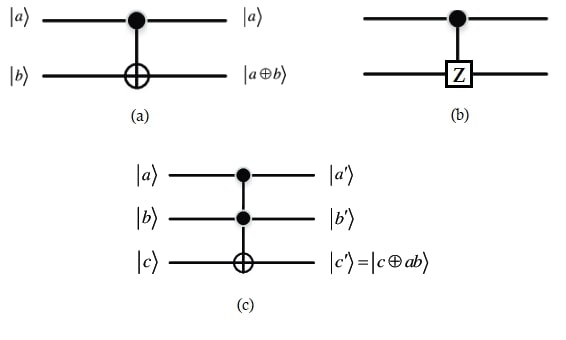}
  \caption {Basic quantum gates. (a) Controlled $NOT$ gate (b) Controlled $PHASE$ gate (c) Controlled-Controlled $NOT$ gate}
  \label{fig:basicgates}
\end{figure}
\begin{figure}
  \includegraphics[width=8cm , height=4cm]{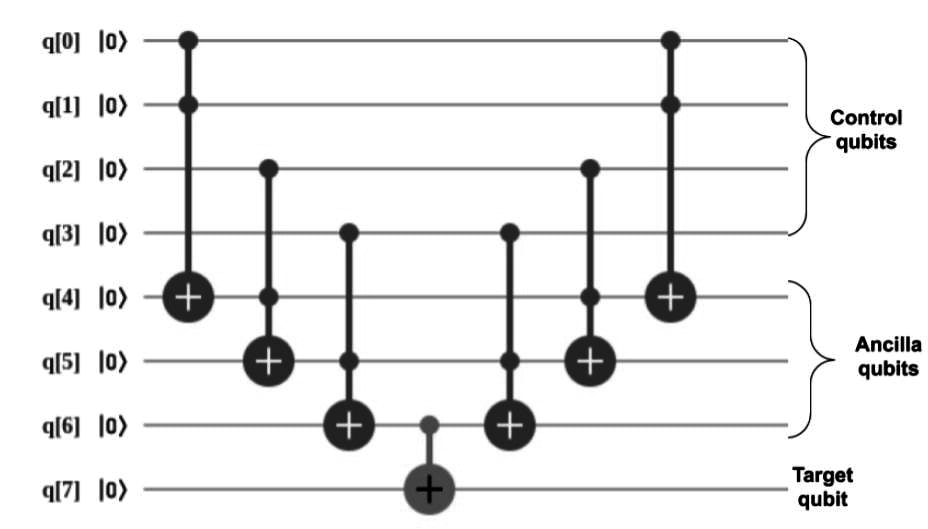}
  \caption{Decomposition of a 4 qubit multi controlled tofolli gate in $CNOT$ and $C^2NOT$ gates}
  \label{fig:tofolli4}
\end{figure}
Measuring a quantum state causes disturbance in quantum mechanical system by resulting in degeneration of superposed quantum state and its convergence into classical state. A collection of measurement operators $\{M_K\}$ where $K$ is a given measurement outcome, is not necessarily unitary. Operators $\{M_K\}$ acting on Hilbert space of the given state satisfy completeness equation, $\sum_KM_K^+M_K=I$. For a quantum state $\phi$, the probability of obtaining the measurement outcome $m$ is $P(m)=\braket{\phi|M_m^+M_m|\phi}$ and the resulting quantum state is $(\braket{\phi|M_m^+M_m|\phi})^{-\frac{1}{2}}*M_m\ket{\phi}$. Completeness equation encodes the fact that measurement probabilities over all the outcomes sum to unity.
A quantum ORACLE is black box representation of a quantum circuit which acts as subroutine of a quantum algorithm. Input to the ORACLE is a boolean function $f$, such that $f: \{0,1\}^n \to \{0,1\}^n$. Function $f$ is said to be queried via an ORACLE $O_f$ where, $\ket{x}\ket{q} \to \ket{x}\ket{q \oplus f(x)}$, $\vdash x \in \{0,1\}^n$ and $q \in \{0,1\}^m$. The above mapping can be implemented by an UNITARY circuit $U_f$ of the form:
\begin{equation}
U_f=\sum_{x \in \{0,1\}^n}\sum_{q \in \{0,1\}^m}\ket{x}\braket{x|X|q \oplus f(x)}\bra{q}
\end{equation}
Effect of ORACLE needs to be determined in all basis states.
\section{Existing notion of QACO and their lack of universality}
In ACO, complexity of exploring the possible paths from food source to nest and exploiting the whole search space increase significantly with the increase in number of paths and number of iterations to achieve better performance due to sequential mode of execution of the algorithms. Quantum parallelization and quantum state entanglement can substantially reduce the algorithmic complexity involved in exploration of large solution space of optimization problems. A fault-tolerant quantum computer with $50$ error corrected qubits can encode $2^{50}$ number of paths simultaneously, which takes $2^{50}$ number of bits in its classical counterpart.
\subsection{Quantum-inspired Evolutionary Algorithm}
Quantum-inspired Evolutionary Algorithms (QEA) fed by probabilistic mechanism of quantum computation, have been applied in the existing research of QACO problems. The smallest information unit in QEA is Q-bit, defined as $[\alpha,\beta]^{T}$. \cite{wang-qa} \cite{wang2007} $\alpha$ and $\beta$ represent complex numbers to satisfy the normalization condition $|\alpha|^{2}+|\beta|^{2}=1$. By a process of probabilistic observation, each Q-bit can be rendered into one binary bit. A Q-bit representation, employing a Q-bit to describe a probabilistic linear superposition can be extended to a multi Q-bit system as shown in equation \ref{multiQ-bit}.
\begin{equation}
Q=\begin{bmatrix}
\begin{matrix}
\alpha_{1} \\ \beta_{1}
\end{matrix}
& | &
\begin{matrix}
\alpha_{2} \\ \beta_{2}
\end{matrix}
& | \cdots & | &
\begin{matrix}
\alpha_{m} \\ \beta_{m}
\end{matrix}
\end{bmatrix}
\label{multiQ-bit}
\end{equation}
\begin{equation}
Q=\begin{bmatrix}
\begin{matrix}
\frac{-\sqrt{3}}{3} \\ \frac{\sqrt{6}}{3}
\end{matrix}
& | &
\begin{matrix}
\frac{\sqrt{2}}{3} \\ \frac{\sqrt{7}}{3}
\end{matrix}
& | &
\begin{matrix}
\frac{-\sqrt{5}}{3} \\ \frac{-2}{3}
\end{matrix}
\end{bmatrix}
\label{multiQ-bit-ex}
\end{equation}
The above example in equation \ref{multiQ-bit-ex} represents a linear probabilistic superposition of $2^{3}=8$ states as $\ket{000}$, $\ket{001}$, $\ket{010}$, $\ket{011}$, $\ket{100}$, $\ket{101}$, $\ket{110}$ and $\ket{111}$, where its superposed state can be described as:
\begin{multline}
\ket{\psi} = \frac{\sqrt{30}}{27}\ket{000}+\frac{\sqrt{24}}{27}\ket{001}-\frac{\sqrt{105}}{27}\ket{010}+\frac{\sqrt{84}}{27}\ket{011} \\ +\frac{\sqrt{60}}{27}\ket{100}-\frac{\sqrt{48}}{27}\ket{101}+\frac{\sqrt{210}}{27}\ket{110}-\frac{\sqrt{168}}{27}\ket{111}
\end{multline}
A conventional binary solution is constructed through Q-bit observation, where for a bit $r_{i}$ of a binary individual $r$, a chosen value of random number $\eta \in [0,1]$ is compared with $\alpha_{i}$ of Q-bit individual $P$. \cite{wang-qa} The binary encoding process is as follows:
\begin{equation}
    \begin{cases}
      r_{i}=0, & \text{if $|\alpha_i|^{2}>\eta$}\\
      r_{i}=1, & \text{if $|\alpha_i|^{2}\leq\eta$}\\
    \end{cases}
\end{equation}
Generation step is followed by fitness evaluation step and its outcome is then processed through a quantum rotation gate $R(\theta)$ operating as follows:
\begin{equation}
\begin{bmatrix}
\alpha_{i} \\ \beta_{i}
\end{bmatrix}^{'}
= R(\theta_{i})\begin{bmatrix}
\alpha_{i} \\ \beta_{i}
\end{bmatrix}
=
\begin{bmatrix}
cos \theta_{i} & -sin \theta_{i} \\ 
sin \theta_{i} & cos \theta_{i}
\end{bmatrix}
\begin{bmatrix}
\alpha_{i} \\ \beta_{i}
\end{bmatrix}
\end{equation}
In order to converge to fitter states, quantum rotation gate is updated.\cite{wang2007} Rotating angle $\theta_{i}$ has a huge significance in performance of Quantum-inspired EAs, where $\theta_{i}$ can be defined as follows:
\begin{equation}
\theta_{i}=sign(\alpha_{i},\beta_{i})\Delta\theta_{i}
\end{equation} 

$sign(\alpha_{i},\beta_{i})$ represents sign of rotating angle $\theta_{i}$ to determine the direction. The value of $sign(\alpha_{i},\beta_{i})$ and $\Delta\theta_{i}$ are decided by looking up into a table in order to compare performance between the solution provided by current individual and the best solution obtained so far. \cite{wang-qa} \cite{wang2007} 
\subsection{Lack of universality of search process of QEA}
The main disadvantage appearing during the process of quantum state rotation lies in the dependency of using a lookup table for fixing the quantum rotation angle. Fixed rotating angle causes negative impact on search speed for an adaptive network, hence can limit the universality of search process by a significant extent. Enhancement of local searching ability and finding escape from local optima might be of great challenge in an ineffective rotating angle updation strategy. Moreover, QEA is not a quantum algorithm, rather it is an evolutionary algorithm inspired by quantum. In the next three sections we have proposed, illustrated and analyzed a novel quantum algorithm for ant colony optimization which is solely based on iteration driven path selection and convergence to the path having maximum pheromone.
\section{MNDAS Algorithm for Quantum Ant Colony Optimization}
Our proposed Ant Colony Optimization algorithm $MNDAS$ is the quantum version of the very basic ant colony system, where the problem space comprises of several  number of parallel paths existing between food source and Ant colony. We assume the number of ants starting from food source to ant colony in a given period of time is evenly distributed with respect to time.  Procedure begins with quantum state preparation to encode all possible paths encountered by ants while traversing from food source to colony. The encoded paths will undergo uniform superposition in order to be initially selected by the quantum ants with equal probability. Additionally, $MNDAS$ algorithm presumes that there is no pheromone already deposited in any path prior to execution.
\begin{figure*}
  \includegraphics[width=\linewidth , height=20cm]{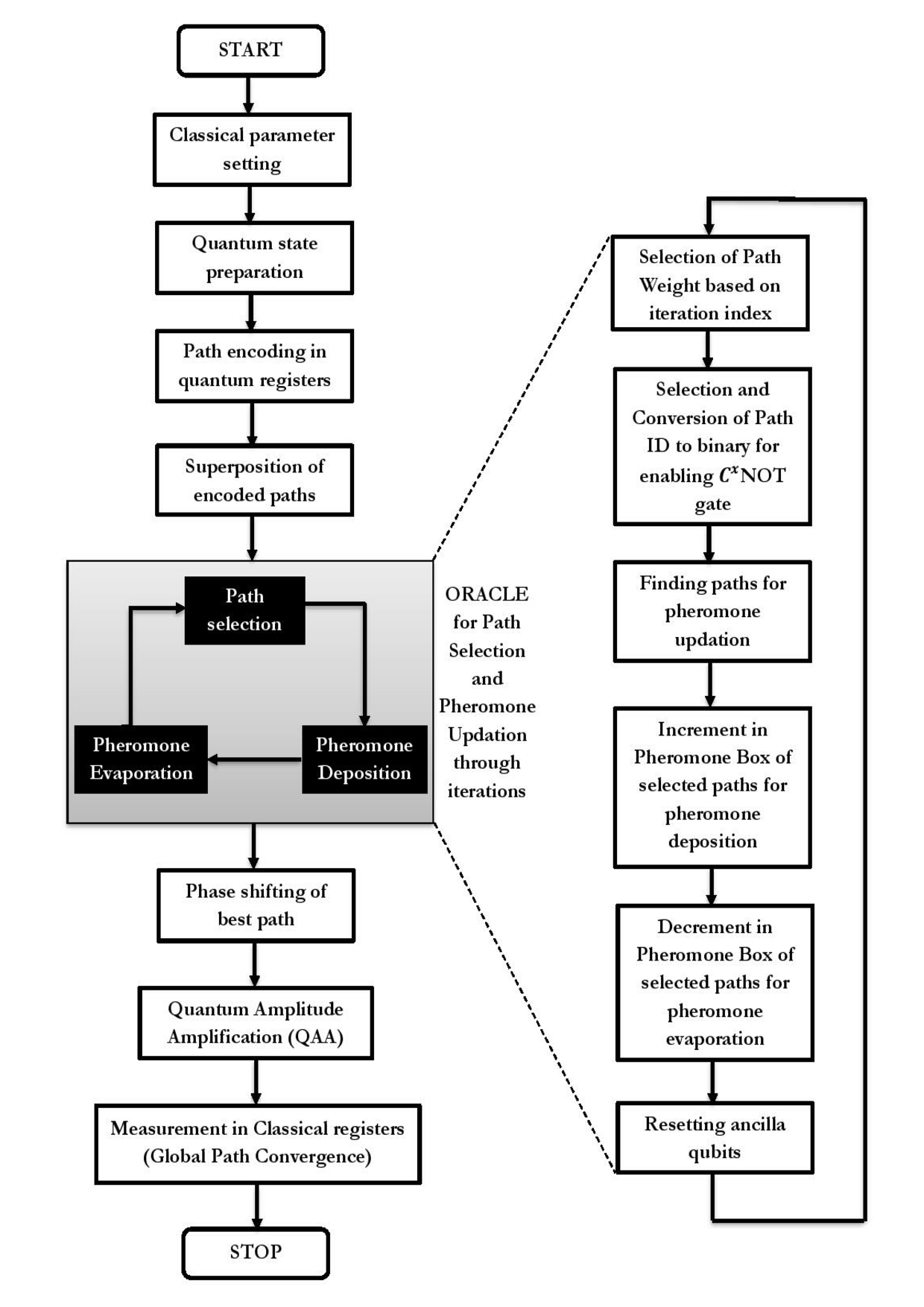}
  \caption{Execution flow of $MNDAS$ algorithm}
  \label{fig:flow}
\end{figure*}

An ORACLE function has been introduced for selecting paths and updating pheromone through multiple iterations. Generally, pheromone deposition and pheromone evaporation will take place for selected and unselected paths respectively in each iteration. Pheromone updation is restricted for a path where convergence criterion has already been met. The ‘best-path’ obtained through a sufficient number of iterations will automatically contain maximum amount of deposited pheromone. Once the solution is reached, our procedure will undergo a phase shift and an amplitude amplification to identify the path from the initial superposition of paths which further will be measured in classical registers. Figure \ref{fig:flow} shows a detailed control flow of the execution of $MNDAS$ algorithm.
Procedure $MNDAS()$ begins with quantum and classical register initialization for ant colony optimization problem. It encodes all the paths and pheromone in quantum registers through $init\_Ant()$ procedure. {$ Q[p_0], Q[p_1], ..., Q[p_{x-1}] $} are the $x$ qubits for all encoded paths to colony. A pheromone box with $d$ number of qubits, {$ Q[ph_0], Q[ph_1], ..., Q[ph_{d-1}] $} is introduced in $init\_Ant()$ to keep track of pheromone density distribution of each path during pheromone updation. Three additional qubits have been initialized with two ancillas $Q[a_1]$ and $Q[a_2]$ representing temporary qureg contents and $Q[a_{target}]$ for target qubit. $init\_Ant() $ also sets number of iteration for convergence ($K$), total number of paths ($n$) and all the path weights from food source to colony ($W$). 
\begin{algorithm}
\SetAlgoLined
\KwResult{Initialisation and parameter setting}
 $\triangleright$ Setting classical parameters \\
 $K \leftarrow$ Constant to denote number of iterations\;
 $n \leftarrow$ Number of paths\;
 $W[n] \leftarrow$ Path weights\;
 $x \leftarrow \lceil \log_2n \rceil$ Number of qubits for path encoding\;
 $Q[p_0 .. p_{x-1}] \leftarrow $ Respective qubits for path encoding\;
 $Q[a_1],Q[a_2],Q[a_{target}] \leftarrow$ Ancilla qubits\;
 $d \leftarrow$ Number of qubits to encode pheromone distribution\;
 $Q[ph_0 .. ph_{d-1}] \leftarrow $ Respective qubits to encode pheromone distribution\;
 $C[0..x-1] \leftarrow$ Number of classical registers for measurement\;
$\triangleright$ Initializing qubits \\
Set $qureg$ $Q[p_0..p_{x-1},a_1,a_2,ph_0..ph_{d-1},a_{target]}$ as $\ket{0}$\;
$\triangleright$ Quantum superposition of $x$ number of encoded paths' qubits. \\
$H(Q[p_0..p_{x-1}])$\;
 \caption{$init\_Ant()$}
 \label{initant}
\end{algorithm}
A total of $(x+d+3)$ entangled qubits are initialized with $\ket{0}$ of which first $x$ qubits undergo quantum superposition with the help of $x$ single qubit Hadamard gates.
\begin{algorithm}
\SetAlgoLined
\KwResult{Possible path exploration by Quantum Ants}
\While{$i$ in $(0, n-1)$}{
  $\triangleright$ Selecting paths corresponding to the iteration \\
  \If{$t \% W[i] == 0$}{
   \While{$l$ in $(p_0, p_{x-1})$}
   {
   \If{$i \% 2 == 0$}{
   $NOT(Q[l])$\;
   }
   }
   $\triangleright$ Multi Controlled Toffoli implementation to raise ancilla for selected paths  \\
   $C^{x}NOT(Q[p_0..p_{x-1}], Q[a_1])$\;
   $C^{x}NOT(Q[p_0..p_{x-1}], Q[a_2])$\;
   $\triangleright$ Reversible operation to achieve initial path encoding \\
   \While{$l$ in $(p_0, p_{x-1})$}
   {
   \If{$i \% 2 == 0$}{
   $NOT(Q[l])$\;
   }
   }
   }
  }\
  $NOT(Q[a_2])$\;
  $update\_Pheromone()$\;
  $\triangleright$ Resetting ancilla after pheromone box updation for each path \\
  $RESET(Q[a_1])$\;
  $RESET(Q[a_2])$\;
 \caption{$ant\_Execute(t)$}
 \label{antexecute}
\end{algorithm}
The initialization step along with problem encoding is followed by an iterative $ORACLE$. The $ORACLE$ consists of procedure $ant\_Execute()$, which performs path selection by picking up the indices of currently explored paths with the help of $MCT$ gates and $update\_Pheromone()$ to update the pheromone box. In order to encode a total of $16$ paths with $x=4$ and $d=4$ to implement corresponding $ant\_Execute()$ for path index $i \neq 11 ... 1$, $NOT$ gates ($X$) are used in respective $0$ positions to enable $MCT\ (C^xNOT)$ gate as shown in figure \ref{fig:antCircuit}. Implementation phase of $ant\_Execute()$ necessitates the decomposition of $MCT$ gates into $CNOT$ and $TOFFOLI\ (C^2NOT)$ as shown in figure \ref{fig:tofolli4}. In each iteration performed in procedure $ant\_Execute()$, ancilla qubits are raised for all the selected paths to be identified during $update\_Pheromone()$. Uncompute task has been performed on path encoding qubits to get back their initial setting and moreover, the ancilla qubits are reset at the end of each iteration of algorithm \ref{antexecute}. \textit{"iteration of ant\_Execute(t)"} section of the circuit shown in figure \ref{fig:antCircuit} is an instance of a single iteration. The circuit will be expanded in that section only along with the corresponding iteration steps. Then, selection of best path by the system, quantum amplitude amplification and measurement at the classical register are done at the end of the whole circuit as shown in figure \ref{fig:antCircuit}.
\begin{algorithm}
\SetAlgoLined
\KwResult{Updating pheromone density based on selected paths}
 $\triangleright$ Avoiding pheromone deposition for the path which has pheromone density as 111 ... 11. i.e. pheromone box for that path is full.  \\
 $C^{d}NOT(Q[ph_0..ph_{d-1}], Q[a_1])$\;
 $pheromone\_Deposition()$\;
 $C^{d}NOT(Q[ph_0..ph_{d-1}], Q[a_2])$\;
 $NOT(Q[ph_0..ph_{d-1})]$\;
  $\triangleright$ Avoiding pheromone evaporation for the paths which has pheromone density as 000 ... 00. i.e. pheromone box for those paths are empty.  \\
 $C^{d}NOT(Q[ph_0..ph_{d-1}], Q[a_2])$\;
 $NOT(Q[ph_0..ph_{d-1}])$\;
 $pheromone\_Evaporation()$\;
\caption{$update\_Pheromone()$}
\label{updatepheromone}
\end{algorithm}
\begin{figure*}
  \includegraphics[width=\linewidth]{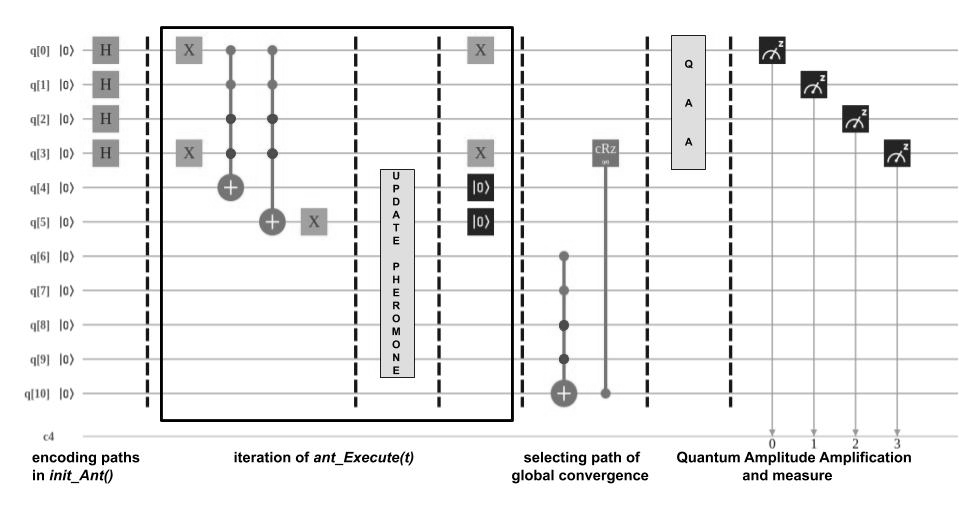}
  \caption{Quantum circuit synthesis by MNDAS algorithm for 6 as a selected path in an anonymous iteration}
  \label{fig:antCircuit}
\end{figure*}
\begin{figure}
  \includegraphics[width=8cm , height=3.5cm]{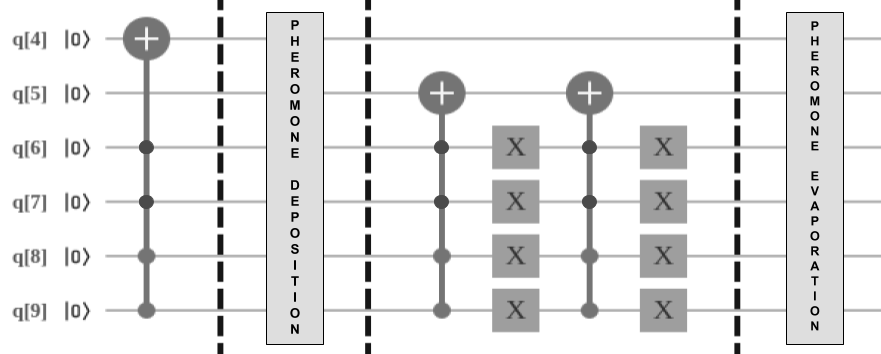}
  \caption{Quantum circuit for $update\_Pheromone()$}
  \label{fig:updatePheromone}
\end{figure}

\begin{figure}
  \includegraphics[width=8.5cm , height=12cm]{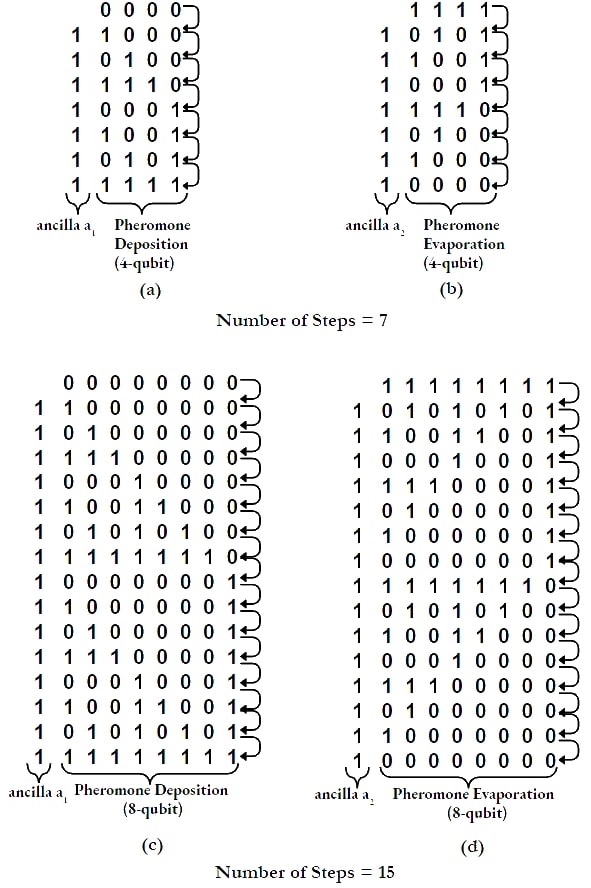}
  \caption{Example of a $4$ qubit and $8$ qubit pheromone box with (a), (c) representing pheromone deposition on selected paths and (b), (d) representing pheromone evaporation from unselected paths.}
  \label{fig:pheromonebox}
\end{figure}
\begin{algorithm}
\SetAlgoLined
\KwResult{Increase in pheromone density of selected path by unit}
 set $r \leftarrow ph_d-2$\;
 \While{$m$ in $(r, ph_0)$}
   {
   $C^{2}NOT(Q[a_1], Q[m], Q[m+1])$\;
   }
 $CNOT(Q[a_1], Q[m])$\;
\caption{$pheromone\_Deposition()$}
\label{pheromonedeposition}
\end{algorithm}
\begin{algorithm}
\SetAlgoLined
\KwResult{Decrease in pheromone density for the paths by unit excluding 
$i)\ Selected\ path$, $ii)\ Paths\ having\ pheromone\ density\ of\ 11...1$ and $iii)\ Paths\ having\ pheromone\ density\ of\ 00...0$}
 set $r \leftarrow ph_d-2$\;
 $CNOT(Q[a_2], Q[ph_0])$\;
 \While{$m$ in $(ph_0, r)$}
   {
   $C^{2}NOT(Q[a_2],Q[m], Q[m+1])$\;
   }
\caption{$pheromone\_Evaporation()$}
\label{pheromoneevaporation}
\end{algorithm}
\begin{algorithm}
\SetAlgoLined
\KwResult{Global path convergence of Ant Colony}
 $\triangleright$ Qubits initialization \\ 
 $init\_Ant()$\;
 \While{$t$ in $(1,K)$}{
  $\triangleright$ Iteration for path selection and pheromone updation \\
  $ant\_Execute(t)$\;
  }
$\triangleright$ Selecting best path \\ 
$C^{d}NOT(Q[ph_0..ph_{d-1}], Q[a_{target}])$\;
$\triangleright$ Phase shifting of target for best obtained path \\ 
$CPHASE(Q[a_{target}], Q[P_{x-1}], \pi)$\;
$\triangleright$ Quantum Amplitude Amplification \\ 
$QAA(Q[p_0..p_{x-1}])$\;
$\triangleright$ Measuring global path convergence for ant colony  \\ 
measure $Q[p_0..p_{x-1}]$ to $C[0..x-1]$\;
 \caption{$MNDAS()$}
 \label{mndas}
\end{algorithm}
Pheromone distribution among all paths is illustrated in the procedure $update\_Pheromone()$. Figure \ref{fig:updatePheromone} shows quantum ants perform $pheromone\_Deposition()$ on the selected paths chosen earlier by procedure $ant\_Execute()$ and $pheromone\_Evaporation()$ for the unselected paths. The shorter paths are supposed to converge earlier in comparison to comparatively longer paths through iterations leading to global path convergence. Number of qubits chosen for encoding the pheromone box $d$ causes variance in convergence time of the ant colony optimization problem. For a $4$ qubit and a $8$ qubit pheromone box, the step-by-step global path convergence is described in figure \ref{fig:pheromonebox} as example. Number of qubits chosen to encode pheromone box affects the overall performance of ACO. Often premature convergence causes lack of adaptiveness into the system in case of link failures or path barriers. The real ant colony behavior through natural synergy and group intelligence allows the ants to choose the second optimal path in case of any obstruction or unreachability in the best path available so far. Our MNDAS QACO algorithm exhibits resemblance to real ant behavior with the help of qubit expansion technique associated with pheromone box updation. If number of qubits in pheromone box is increased, it will undergo more number of iterations and thus will provide a tool for congestion-controlled traffic through controllable duration before global convergence. $pheromone\_Deposition()$ and $pheromone\_Evaporation()$ perform updation in pheromone box by maintaining the push and pop sequence from an unique updation order. For a $4$ qubit pheromone box, the qureg content follows the unique pheromone updation order of $0-8-4-14-1-9-5-15$ as shown in figure \ref{fig:pheromonebox} (a), where $0$ with binary equivalent $0000$ and $15$ with binary equivalent $1111$ represent initial empty pheromone box and box with maximum pheromone respectively. If $d$ represents the total number of qubits to encode pheromone distribution, maximum ($2^{\lfloor \log_2d \rfloor + 1}-1$) number of pheromone distribution states are available in our proposed QACO algorithm before 'best-path' convergence as shown in figure \ref{fig:pboxqubit}.
\begin{theorem}
Second best shortest path will never converge as optimal solution in $MNDAS()$, if there already exists best path with minimum path weight. 
\end{theorem}
\begin{proof}
Each iteration of procedure $ant\_Execute()$ performs path selection by applying modulo division arithmetic of a specific iteration index ($K$) by each of the path weights. If that iteration index is a multiple of path weight, then the path ID of that corresponding path weight is selected for pheromone deposition. Say, the path with minimum weight is $p_{min1}$ and the second best path is $p_{min2}$, where $W_{min1} < W_{min2}$ ($W_{min1}$, $W_{min2}$ are path weights corresponding to paths $p_{min1}$ and $p_{min2}$ respectively). After a sufficient number of iterations ($K_{threshold}$), say, number of times pheromone has been deposited on paths $p_{min1}$ and $p_{min2}$ are $d_1$ and $d_2$ respectively. On the other hand, number of times pheromone has been evaporated from paths $p_{min1}$ and $p_{min2}$ are $e_1$ and $e_2$ respectively. Each deposition indicates selection and each evaporation indicates non-selection of the corresponding path. Now if $r_1$ and $r_2$ denote the number of times paths $p_{min1}$ and $p_{min2}$ will be selected in $ant\_Execute()$ procedure and $K'$ denotes any random iteration index ($K' >> K_{threshold}$), which is multiple of path weights of both $p_{min1}$ and $p_{min2}$,
\begin{equation}
K' = r_1 . W_{min1}
\label{leq1}
\end{equation}
\begin{equation}
K' = r_2 . W_{min2}
\label{leq2}
\end{equation}
From equations \ref{leq1} and \ref{leq2}, we get
\begin{equation}
\begin{aligned}
(W_{min1} < W_{min2}) \implies (r_1 > r_2) \\ \implies ((d_1 > d_2)\wedge (e_1 < e_2))
\end{aligned}
\label{leq3}
\end{equation}
Since, the shortest path with minimum path weight will always undergo selection more than any other path present in the graph (including the second best path), our algorithm guarantees convergence of the shortest path as the optimal solution.
\end{proof}
\begin{theorem}
If the shortest path is removed from the search space before convergence, then our solution converges to the next best path after adequate iterations.
\end{theorem}
\begin{proof}
Un-reachability in the best path can be mathematically mapped as the path with infinite path weight (or a path weight with very large magnitude). Since, our algorithm performs pheromone box updation dynamically, it will adapt to the new changes which have taken place in the system. Due to being unselected in all the iterations taken place after the previous shortest path has attained an infinite path weight, the said path will eventually undergo several evaporations and being at all $0's$ corresponding to its entry in the pheromone box. On the other hand, the second optimal path will have the chance to be selected maximum number of times during procedure $ant\_Execute()$.
\\$n$ = total number of paths present from food source to ant colony.
\\$P$ = set of all paths.
\\$p_m$ = second optimal path (which currently is the best path with the minimum weight due to absence of the shortest path).
\\$p_i$ = any arbitrary path other than $p_m$.
\\$W_m$, $W_i$ = path weights corresponding to $p_m$ and $p_i$ respectively.
\\$r_m$, $r_i$ = number of times paths $p_m$ and $p_i$ will be selected respectively.
\\$\Delta_m$, $\Delta_i$ = time taken for convergence in pheromone box by $p_m$ and $p_i$ respectively.
\begin{equation}
\begin{aligned}
\exists p_m \in P \forall (p_i \in P)\wedge(0 \le i \le n-1) ((W_m < W_i) \\ \implies (r_m > r_i)\wedge(\Delta_m < \Delta_i))
\end{aligned}
\end{equation}
\end{proof}
Whenever $update\_Pheromone()$ procedure is invoked, it typically puts a constraint on amount of pheromone deposition based on the pheromone density in pheromone box of the corresponding path. Paths with pheromone density $11...1$ is supposed to contain maximum amount of pheromone. In such cases, $pheromone\_Deposition()$ procedure restricts itself to increase the amount of deposited pheromone by changing the ancilla qubit $Q[a_1]$ as further deposition of pheromone will cause disturbance in convergence by resetting the pheromone density to $00...0$. On the contrary, procedure $pheromone\_Evaporation()$ does work only for the unselected paths of a particular iteration. The paths with pheromone density $00...0$ do not even contain any pheromone. $pheromone\_Evaporation()$ has to check an underflow condition by putting a barrier into pheromone to be evaporated from a path without any available pheromone deposited on it. There are other paths which might be unselected in $i^{th}$ iteration having pheromone density of $11...1$. In order to prevent pheromone evaporation from an already convergent path, $pheromone\_Evaporation()$ also does not decrease the pheromone density of such paths as well since it will affect global convergence in QACO.
\\
\begin{figure}
  \includegraphics[width=8cm, height=6cm]{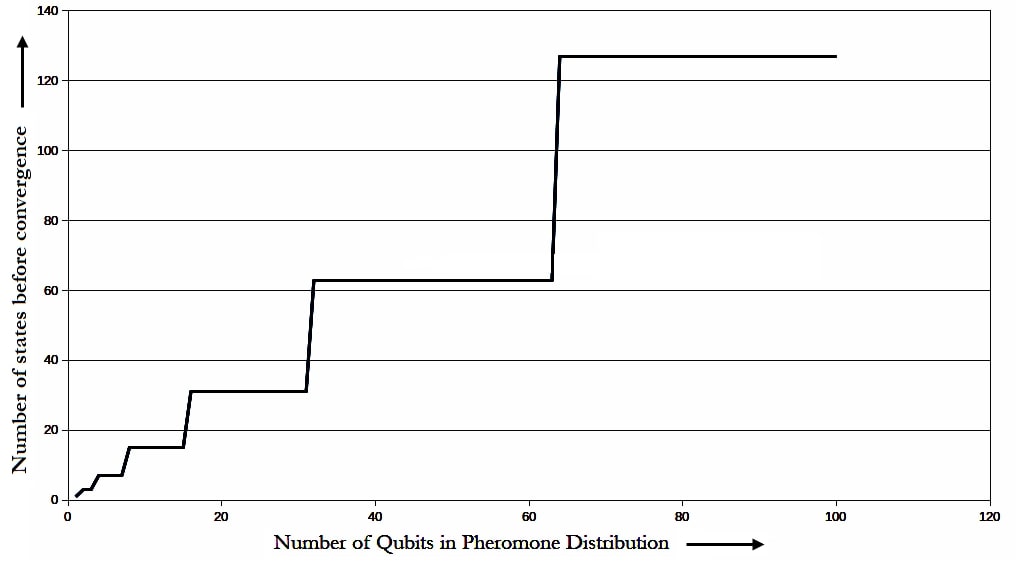}
  \caption{Number of maximum pheromone states in pheromone box with respect to number of qubits. }
  \label{fig:pboxqubit}
\end{figure}
After complete exploration of all possible paths by quantum ants, QACO will converge to optimal state with shortest path chosen as output of procedure $ant\_Execute()$. The index value of the shortest path will be selected by using a MCT gate with pheromone box {$Q[ph_0, ph_1, .... ph_{d-1}]$} as control and $Q[a_{target}]$ as target. The MNDAS algorithm uses a controlled phase shift ($CPHASE$) gate performing a $\pi$ rotation of the qubit in most significant position ($Q[p_{x-1}]$) based on value of $Q[a_{target}]$. Once phase is shifted for the best path, procedure $MNDAS()$ invokes Quantum Amplitude Amplification (QAA) technique in order to amplify the probability density value of the path with minimum weight (shortest path from food source to colony). \cite{grover} QAA technique is followed by a measurement step where convergence of a quantum superposed state is mapped into a classical register.
The concern associated with most of the combinatorial optimization problems is convergence. Stochastic search procedures like classical ant colony optimization face challenges in achieving optimality in solution obtained, as pheromone update often prevents an algorithm to reach optimal state. \cite{gutjhar} \cite{hoos} It is worth mentioning that our proposed MNDAS QACO algorithm is well suited in achieving convergence in value as well as convergence in solution. Our algorithm yields optimal solution atleast once and optimality is preserved in the same solution with course of time-variant iterations; thus ensuring convergence in value and convergence in solution both.
\section{Algorithm Analysis and Result}
In real ant colonies, ants aim to find the shortest path from a colony to food source. Since, ants deposit a certain amount of pheromone in its path from nest to food and while making the return trip, follow the same path marked previously along with depositing pheromone on the same, ants following the shorter path are expected to return earlier. The real key of our MNDAS algorithm follows the same principle  where the rate of deposition of pheromone has been made faster on the shorter path in comparison to the longer paths to induce pheromone evaporation effect.
\begin{figure}
  \includegraphics[width=8cm, height=5cm]{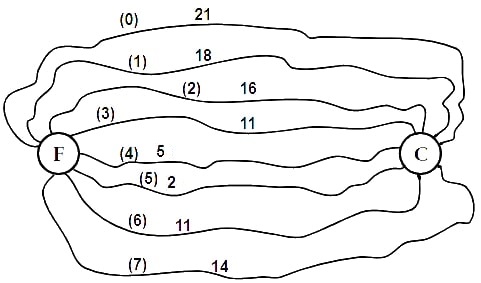}
  \caption{A example of Simple Ant colony optimization with 8 different paths, encoded in a 3 qubit quantum system}
  \label{fig:example}
\end{figure}
Pheromone evaporation takes place periodically by a certain amount at a constant rate which implies the existence of frequently visited paths only through pheromone deposition as rarely visited paths by ants will undergo accelerated evaporation followed by no existence due to lack of pheromone deposition. All ants starting their food searching journey can learn from the information left by previously visited ants and can get guidance to follow the shorter path directed by maximum pheromone deposit.
This foraging behavior of real ants can be mapped into shortest path finding problem where a number of artificial ants mimicking the data packets will build solutions and exchange relevant information on the quality of the solutions via a communication scheme which is expressed in our algorithm as $update\_Pheromone()$ consisting of $pheromone\_Deposition()$ and $pheromone\_Evaporation()$ quantumly. In order to elucidate our proposed $MNDAS$ algorithm, we have considered a graph example with $8$ possible paths existing between food source and ant colony as shown in figure \ref{fig:example}. 
\begin{figure*}
  \includegraphics[width=\linewidth , height=19cm]{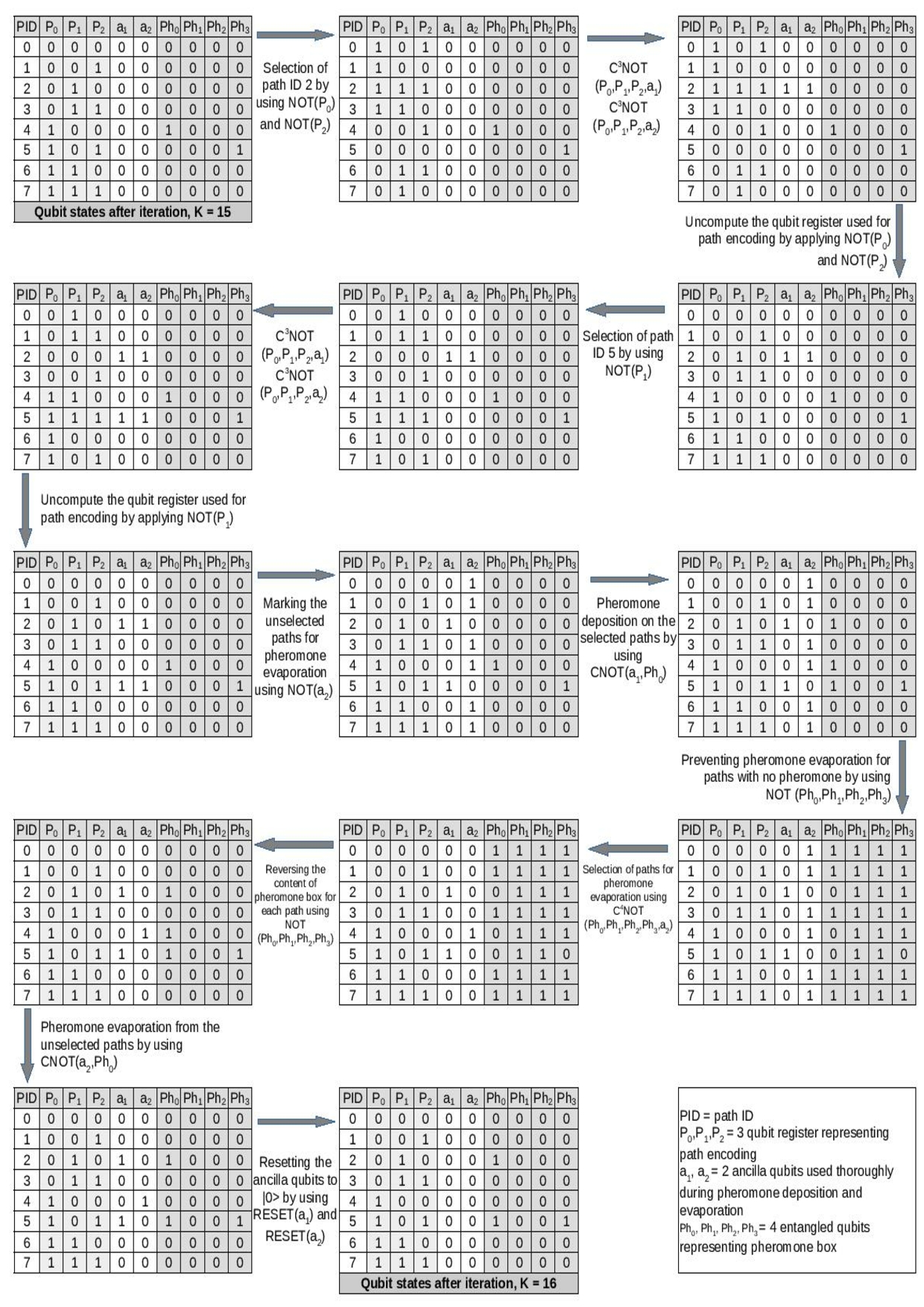}
  \caption{Change in qubit states of the pheromone box during a single iteration for $K = 16$. Qubit $a_{target}$ is not shown as it is unaffected during quantum iteration.}
  \label{fig:iteration}
\end{figure*}

\begin{figure}
  \includegraphics[width=9cm , height = 25cm]{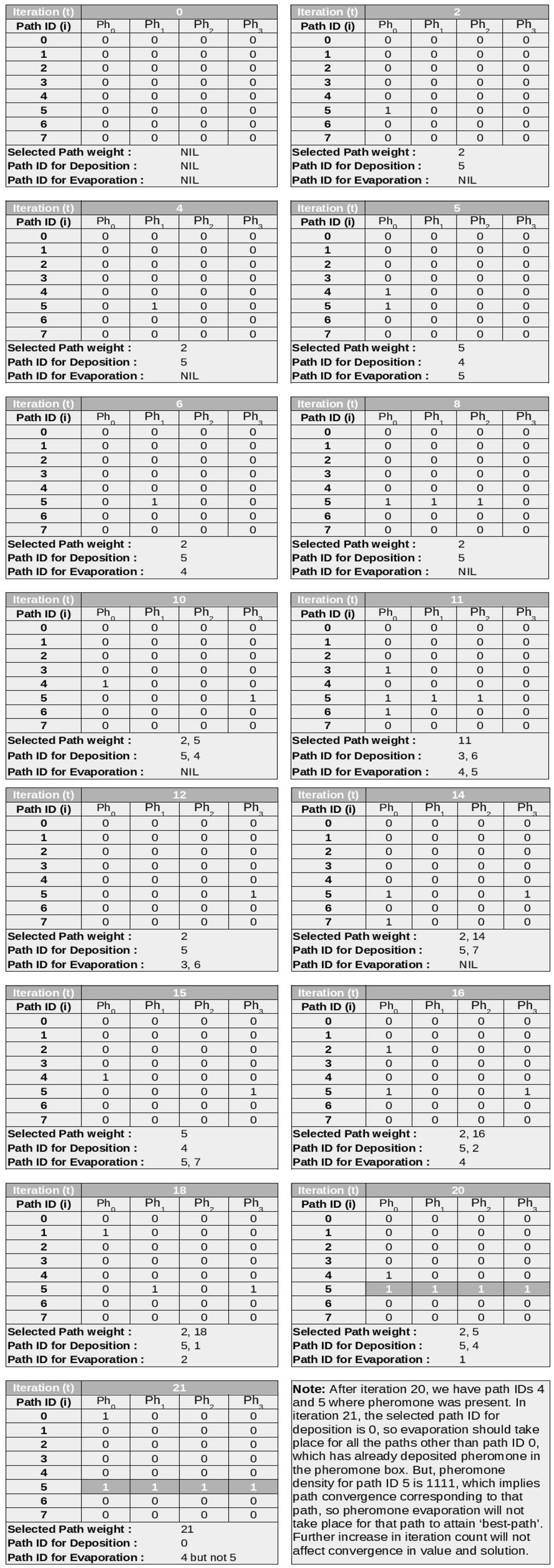}
  \caption{Qubit state status of pheromone box for the iteration $K = 0,2,4,5,6,8,10,11,12,14,15,16,18,20,21$}
  \label{fig:pheromonestates}
\end{figure}

\begin{table}
\begin{center}
\begin{tabular}{|c|c|c|c|c|c|c|c|c|}\hline
$path\ id$ & $0$ & $1$ & $2$ & $3$ & $4$ & $5$ & $6$ & $7$ \\
\hline
$Weight$ & $21$ & $18$ & $16$ & $11$ & $5$ & $2$ & $11$ & $14$\\
\hline	
\end{tabular}
\end{center}
\caption{Cost corresponding to paths in a 3 qubit QACO as shown in Figure \ref{fig:example}}
\label{p-w-3}
\end{table}
In $init\_Ant()$ procedure, all the classical and quantum parameter settings have been performed. For a path count of $n = 8$, the number of qubits representing path encoding is $x = 3$. If we use a $4$ qubit pheromone box ($d = 4$), the initial quantum superposed state $\ket{\psi}$ with ($x+d+3$) qubits can be denoted as equation \ref{eq:state}, where all the qubits except the path encoding qubits are initialized with $\ket{0}$.
\begin{equation}
\begin{aligned}
\ket{\psi} = 1/\sqrt{8}\ket{0000000000} + 1/\sqrt{8}\ket{0010000000} + \\	 1/\sqrt{8}\ket{0100000000} + 1/\sqrt{8}\ket{0110000000} + \\
1/\sqrt{8}\ket{1000000000} + 1/\sqrt{8}\ket{1010000000} + \\
1/\sqrt{8}\ket{1100000000} + 1/\sqrt{8}\ket{1110000000} \\
\end{aligned}
\label{eq:state}
\end{equation}
Equation \ref{eq:state} also ensures uniform distribution of probability through quantum superposition to achieve equiprobable selection chances of all possible paths by quantum ants.
The procedure $ant\_Execute()$ performs multiple iterations for path selection and pheromone updation to achieve final convergence to best-path. Each iteration of $ant\_Execute()$ will necessitate the pheromone box to undergo successive changes in qubit states. A detailed transition showing the changes in pheromone box content from the end of iteration $K=15$ to the end of iteration $K=16$ has been shown in figure \ref{fig:iteration}.
In each iteration, the path IDs selected for pheromone deposition and pheromone evaporation have been explicitly shown. For example as shown in figure \ref{fig:iteration} during iteration $K=16$, path IDs $2$ and $5$ are selected for deposition as iteration index is a multiple of their path weights $16$ and $2$ respectively. The path with IDs $0, 1, 3, 6$ and $7$ contain all $0$s in pheromone box implying their inapplicability for pheromone evaporation except the path with ID $4$ which is unselected in that iteration.
The path encoding qubits get back to their initial configuration due to uncompute operation specified in $ant\_Execute()$. Moreover, the two ancilla qubits $Q[a_1]$ and $Q[a_2]$ will also be resetted at the end of each iteration. $Q[a_{target}]$ qubit is not involved in the whole iterative procedure and hence, holds the initial value. So, in each iteration only the state of the qubits $Q[Ph_0, Ph_1, Ph_2, Ph_3]$ representing pheromone box will be updated. Figure \ref{fig:pheromonestates} represents the status of qubit states of pheromone box in multiple iterations for $K=0, 2, 4, 5, 6, 8, 10, 11, 12, 14, 15, 16, 18, 20, 21$. In iteration $K=21$, the qubit register for pheromone box contains all $1$s corresponding to path ID $5$. The path with ID $5$ denotes the shortest path corresponding to the minimum path weight $2$ as shown in figure \ref{fig:pheromonestates}.
The function $ant\_Execute()$ is the input to the ORACLE which yields the shortest path $5$ as output of the ORACLE. Controlled phase shift gate performs a $\pi$-rotation of path index of the selected shortest path enabling the amplitude  of the selected path to be phase shifted by $\pi$ to undergo Quantum Amplitude Amplification (QAA). QAA amplifies the probability amplitude of the selected path as shown in figure \ref{fig:exsol} by applying phase inversion followed by performing inversion about mean operation on target qubit $Q[a_{target}]$. $MNDAS$ QACO algorithm is free from earlier convergence to local optima and lack of universality of search space with the help of adaptive quantum algorithms, \ref{pheromonedeposition} and \ref{pheromoneevaporation}.
\begin{figure}
  \includegraphics[width=8cm, height=5cm]{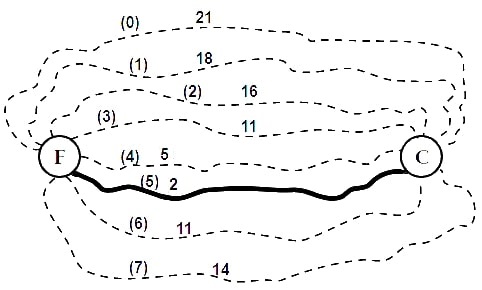}
  \caption{QAA amplifies the probability amplitude of the selected path.}
  \label{fig:exsol}
\end{figure}
\subsection{Results}
We have implemented $MNDAS$ QACO algorithm using IBM QISKIT. The algorithm has been executed on both $15$ qubit IBMQ ($ibmq\_16\_melbourne$) and QASM simulator for $\leq 8$ paths and higher number of paths respectively. Since, $RESET$ operation is currently not supported, we have implemented each step of iteration separately by initializing qubits with the output state of previous iteration. IBMQ error threshold values, $\zeta$ corresponding to a single qubit quantum gate $U_2$ and $CNOT$ gate are depicted as $4.904 e^{-4} \leq \zeta_{U_2} \leq 2.711 e^{-3}$ and $1.250 e^{-2} \leq \zeta_{CNOT} \leq 8.890 e^{-2}$ respectively with average qubit frequency of $4.976$ GHz approximately.
The Ant colony specified in figure \ref{fig:example} shows a simple network with $8$ paths with path ids $0,1,...\:,7$. Encoding $8$ such paths in our proposed $MNDAS()$ algorithm requires $3$ qubits. Among all possible path costs, minimum is $2$ which in turn, is associated with path ID $5$ as shown in figure \ref{fig:exsol}. The outcome of a particular execution of our algorithm on $3$ qubits and number of iterations as $K=200$ is shown in figure \ref{fig:resultq3} where, the path id $5$ undergoes amplitude amplification after being selected as shortest path with probabilistic amplitude $(p_{selected})$ of $0.501$ where $p_{(\forall x\: \in \:n\land\:x \: \neq\:selected)} << p_{selected}$.
\begin{figure}
  \includegraphics[width=6cm, height=4.8cm]{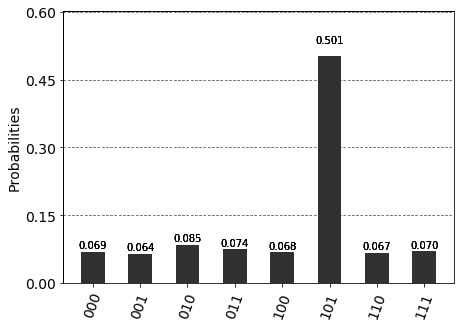}
  \caption{Convergence to the shortest path in a $8$-path ACO}
  \label{fig:resultq3}
\end{figure}
Another example is taken into consideration with $16$ possible paths and $4$ number of qubits to encode all possible paths with path weight $6$ as minimum cost (shortest path) and path id $13$ as shown in figure \ref{fig:resultq4} for the distribution of path weights as specified as table \ref{p-w-4}.
\begin{table}
\begin{center}
\begin{tabular}{|c|c|c|c|c|c|c|c|c|}\hline
$path\ id$ & $0$ & $1$ & $2$ & $3$ & $4$ & $5$ & $6$ & $7$ \\
\hline
$Weight$ & $12$ & $9$ & $24$ & $131$ & $17$ & $99$ & $11$ & $100$\\
\hline	
$path\ id$ & $8$ & $9$ & $10$ & $11$ & $12$ & $13$ & $14$ & $15$ \\
\hline
$Weight$ & $24$ & $31$ & $64$ & $79$ & $73$ & $6$ & $67$ & $101$ \\
\hline
\end{tabular}
\end{center}
\caption{Cost corresponding to paths in a 4 qubit QACO}
\label{p-w-4}
\end{table}
\begin{figure}
  \includegraphics[width=6cm, height=4.8cm]{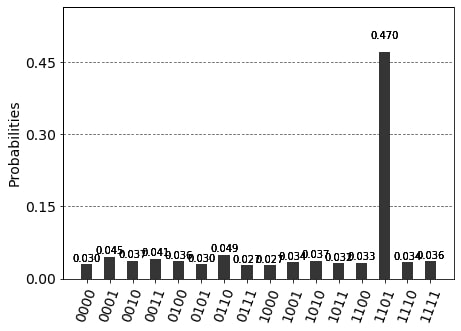}
  \caption{Convergence to the shortest path in a $16$-path ACO}
  \label{fig:resultq4}
\end{figure}
\subsection{Complexity Analysis:}
The expected optimization time of well-known MAX-MIN Ant System (MMAS) for single destination shortest path on a graph G with $m$ number of vertices is $O(m^3 + \frac{m}{\rho})$, where $\rho$ is evaporation factor. \cite{mmhoos} In our proposed QACO algorithm $MNDAS()$, a single quantum superposed state is prepared to encode all possible paths, thus exploration of all the possible $n$ paths can be incorporated in a single iteration in $O(1)$ time. Since procedure $MNDAS()$ has chosen $K$ number of iterations to be performed for global path convergence, the total complexity of $ant\_Execute()$ through $pheromone\_Updation()$ requires $O(K.n) $ time complexity. The QAA performed in $MNDAS$ algorithm incurs a complexity of $O(\sqrt{n})$. So, overall complexity of our proposed novel quantum ACO is $O(K.n+\sqrt{n}) \in O(n)$, as $K$ is constant. Thus, our quantum algorithm provides polynomial speedup over its classical counterpart.
\section{Conclusion}
In this paper, we have proposed a novel quantum algorithm for ant colony optimization to solve computationally hard combinatorial optimization problems. Our algorithm MNDAS QACO (Mrityunjay-Nivedita-Debdeep-Amlan-Subhansu Quantum Ant Colony Optimization) approaches a novel quantum algorithm to be run on a quantum hardware instead of quantum-inspired evolutionary ACO algorithms available so far. Our approach can be modelled as quantum module for a variety of NP-Hard problems namely Travelling Salesman Problem (TSP), Vehicle Routing Problem and Network Routing Problem.\\
MNDAS QACO is an adaptive quantum algorithm ensuring reliability in obtaining the shortest path through $pheromone\_Update()$ quantum module. We have also built up a fault prevention mechanism through structural constraints applied over pheromone deposition and pheromone evaporation to achieve unaffected global convergence of ACO problems. Our future work will incorporate quantum gate cost optimization through fault tolerant logic synthesis of quantum circuits to reduce the gate cost and improve the overall efficiency of our algorithm. In order to physically implement the multi-qubit quantum gates like MCT, Toffoli specified in our algorithms, a considerable number of SWAP gates have to be introduced which in turn, will incur a huge cost over head. This might be our extended research area to work in order to reduce circuit complexity and gate overhead through optimized synthesis of quantum physical circuits. Moreover, we will also focus on enhancing this quantum ACO algorithm for other complex variants of ant colony system.

\bibliographystyle{plain}
\bibliography{sample-base}

\vfill\pagebreak

\section{Appendices}\label{sec14}

\end{document}